\documentclass[11pt,reqno]{amsart}
\usepackage[T1]{fontenc}
\usepackage[utf8]{inputenc}
\usepackage{mathpazo}
\usepackage{amsmath,amssymb,amsthm,mathtools}
\usepackage[margin=1.05in]{geometry}
\usepackage{microtype}
\usepackage{xcolor}
\usepackage{graphicx}
\graphicspath{{figures/}{output/pdf/figures/}}
\usepackage{booktabs}
\usepackage{needspace}
\usepackage{tikz}
\usetikzlibrary{arrows.meta}
\usepackage[colorlinks=true,linkcolor=blue!45!black,citecolor=blue!45!black,urlcolor=blue!45!black]{hyperref}
\hypersetup{pdftitle={Boundary routing and memory in the three-dimensional Manhattan mirror model},pdfauthor={}}
\numberwithin{equation}{section}
\newtheorem{theorem}{Theorem}[section]
\newtheorem{proposition}[theorem]{Proposition}

\theoremstyle{definition}

\theoremstyle{remark}

\newcommand{\Z}{\mathbb Z}

\newcommand{\Pp}{\mathbb P_p}
\newcommand{\Ep}{\mathbb E_p}
\newcommand{\ind}{\mathbf1}

\DeclareMathOperator{\tr}{tr}

\newcommand{\FirstAuthor}{Jian Gu}
\newcommand{\FirstAffiliation}{ESSEC}
\newcommand{\SecondAuthor}{Paul Schaeffer}
\newcommand{\SecondAffiliation}{Polytechnic Institute of Paris}
\newcommand{\ThirdAuthor}{Qin Hao}
\newcommand{\ThirdAffiliation}{Polytechnic Institute of Paris}
\title[Routing and memory in the three-dimensional Manhattan model]{Routing and memory in the three-dimensional mirror model on K-lattice}
\author[\FirstAuthor]{\FirstAuthor}
\address{\FirstAffiliation}
\author[\SecondAuthor]{\SecondAuthor}
\address{\SecondAffiliation}
\author[\ThirdAuthor]{\ThirdAuthor}
\address{\ThirdAffiliation}
\makeatletter
\renewcommand{\@setauthors}{%
  \begingroup\trivlist\centering\footnotesize
  \@topsep30\p@\relax\advance\@topsep by -\baselineskip
  \item\relax
  \begin{minipage}[t]{.30\textwidth}\centering
    {\scshape\FirstAuthor\par}\smallskip
    {\normalfont\footnotesize\FirstAffiliation\par}
  \end{minipage}\hfill
  \begin{minipage}[t]{.30\textwidth}\centering
    {\scshape\SecondAuthor\par}\smallskip
    {\normalfont\footnotesize\SecondAffiliation\par}
  \end{minipage}\hfill
  \begin{minipage}[t]{.30\textwidth}\centering
    {\scshape\ThirdAuthor\par}\smallskip
    {\normalfont\footnotesize\ThirdAffiliation\par}
  \end{minipage}
  \endtrivlist\endgroup}
\renewcommand{\@setaddresses}{}
\makeatother

\subjclass[2020]{60K37, 82B41, 60K35, 05C38}
\keywords{Manhattan mirror model, three-dimensional K-lattice, delocalization, boundary scattering, quenched memory, numerical experiments}

\theoremstyle{definition}
\newtheorem{finding}[theorem]{Numerical finding}
\newcommand{\EE}{\mathbb E}

\begin{document}
\raggedbottom
\begin{abstract}
We study numerical observables relevant to delocalization in the
three-dimensional Manhattan mirror model, also known as the K-lattice.
Our principal experiment separates the locations of transmitting boundary
channels from the routing between them. We derive an exact conditional
comparison that preserves all transmission and reflection aperture sets,
and find a statistically resolved deficit of physical transmission relative
to this comparison. Nearby transmitting channels are themselves suppressed
relative to a uniform subset with the same cardinality. A second experiment
resolves the contribution of repeated vertices by their return age. The
squared norm of the signed old-memory remainder is consistent with an
inverse-square-root decay over the measured range. Annular experiments
quantify decreasing continuation losses, while an initially apparent
history dependence fails independent replication. Finally, a periodic
configuration has a connected auxiliary contact graph but only finite
routing cycles. These results identify estimates on reflection routing and
adaptively selected revisit charges that may assist a rigorous proof.
They do not establish an infinite-volume delocalized phase.
\end{abstract}
\maketitle

\section{Introduction}\label{sec:intro}

The three-dimensional Manhattan mirror model is a deterministic routing
system in an independent random environment. Its underlying graph has four
edges at each vertex, with two incoming and two outgoing directions. Each
vertex carries a single Bernoulli variable: a ray changes axis at a mirror
and continues straight otherwise. Successive visits to the same physical
vertex use the same variable. Consequently, even an individual trajectory
retains information about the environment it has already explored.

Cardy introduced this lattice as a three-dimensional setting in which the
quantum--classical correspondence for class-C networks might clarify the
existence of extended states \cite{Cardy,BeamondCardyChalker}. Subsequent
work calls it the K-lattice and studies its extended and short-loop phases
\cite{Nahum2011,Nahum2013}. We use \(p\) for the mirror, or turning,
probability throughout. This is the complement of the straight probability
used for the K-lattice in those papers. Section~\ref{sec:model} specifies
the graph, orientation, and stopping convention.

Our purpose is to measure quantities that distinguish possible proof
mechanisms. Bulk displacement or a large finite-volume loop alone does not
control how an exposed trajectory interacts with previously visited sites.
Likewise, a transmission count does not describe the full scattering map of
a block. We therefore investigate boundary routing conditional on its
aperture sets, revisit charges resolved by their age, and continuation
between successive spatial scales.

\subsection{Main results}

For two independently sampled slabs, let \(G\) be the number of paths
transmitting from the left exterior face to the right exterior face after
physical gluing. The boundary scattering map of each slab has four blocks,
classified by the input and output faces. Keep the exact domain and image
of every block, and replace its routing by a uniformly random bijection
between these two sets. Let \(H\) be the resulting conditional expected
transmission. This comparison preserves the complete spatial locations of
all transmitting and reflecting ports. Theorem~\ref{thm:fixedsets} gives
an exact formula for \(H\).

\Needspace{7\baselineskip}
\begin{finding}[Transmission with fixed aperture sets]\label{finding:routing}
For \(512\) independent pairs of slabs of width and thickness \(128\),
at mirror probability \(p=0.72\), the empirical means are
\[
 \widehat{\EE G}=1.85938,\qquad
 \widehat{\EE H}=2.06788.
\]
The paired mean of \(H-G\) is \(0.20850\), with Monte Carlo standard
error \(0.03224\). The corresponding deficits at \(p=0.7\) and at
width \(64\) are reported in Table~\ref{tab:gluing}.
\end{finding}

The physical mean in this experiment is approximately \(10.1\%\) below
the comparison mean. Since the terminal sets remain fixed, the comparison
isolates dependence in the reflection routes beyond the positions of the
transmission channels. The rematched boundary permutations are an
artificial comparison; they need not be realizable by local mirrors.

For the second experiment, a finite-state corrector gives the exact stopped
decomposition \(Y_n=M_n-R_n\), where \(M_n\) sums centered charges at
freshly visited vertices and \(R_n\) records charges selected by a second
visit. Write \(Q_n\) for the trace of the predictable quadratic variation
of \(M\). For a first visit at time \(s\) and second visit at time \(t\),
define the kinetic age \(p(t-s)\). The part of \(R_n\) arising from ages
at least \(L\) is denoted by \(R_n^{\ge L}\).

\begin{finding}[The signed old-memory remainder]\label{finding:memory}
Using \(2{,}048\) independent trajectories at each of
\(p=0.025,0.05,0.1,0.2\), the statistic
\[
 B_p(n,L)=\frac{\Ep|R_n^{\ge L}|^2}{p^2\Ep Q_n}
\]
is consistent with an inverse-square-root age tail over the measured
range. At \(n=2^{20}\), finite-horizon corrected exponent fits over
\(L=32,64,\ldots,1024\) are \(0.504,0.500,0.500,0.494\), respectively.
All four whole-trajectory bootstrap intervals contain \(1/2\).
\end{finding}

This observable is the squared norm of a signed, adaptively selected
vector sum. It is not obtained from the number of revisits alone.
Section~\ref{sec:memory} specifies the finite-horizon fit and its
limitations, and Appendix~\ref{app:corrector} proves the corrector
identity used to define the observable.

The following exact result rules out a simple deterministic implication
from connectivity to transport.

\Needspace{7\baselineskip}
\begin{theorem}[A connected contact graph with finite routing]
\label{thm:periodic-obstruction}
There is a periodic configuration of the three-dimensional Manhattan
mirror model for which the contact graph of all-turn hexagons is connected
throughout \(\Z^3\), but every directed trajectory closes after exactly
\(48\) steps.
\end{theorem}

The contact graph and the explicit configuration are constructed in
Appendix~\ref{app:geometry}. The example is deterministic, rather than an
event of positive probability under the iid law. It shows that a
percolation argument must retain information about the routing inside
connected components.

\subsection{Further experiments and organization}

Transmission apertures exhibit a suppression of nearby pairs relative to
a uniform set of the same size. Annular experiments measure the
conditional probability of closure between radii \(R\) and \(2R\).
An exploratory association with first-exit time was tested in two
independent cohorts and did not reproduce its initial strength.

Sections~\ref{sec:gluing}--\ref{sec:annular} present the boundary,
memory, and annular experiments. Section~\ref{sec:methods} explains
sampling and uncertainty estimates. Section~\ref{sec:discussion}
formulates the analytical questions suggested by the data. The exact
comparison formula, corrector identities, contact construction, and
sufficient escape criteria are proved in the appendices. We distinguish
these exact statements from Monte Carlo findings throughout.

\subsection{Relation to earlier work}

The numerical study of localization has its origins in Anderson's model
\cite{Anderson1958} and in the scaling theory of conductance
\cite{AALR}. The three-dimensional Anderson transition provides a
particularly developed setting for finite-size analysis. Early numerical
conductivity calculations treated both two and three dimensions
\cite{KramerMacKinnonWeaire1981}. MacKinnon and
Kramer used localization lengths in long strips and bars to construct
numerical scaling flows \cite{MacKinnonKramer1981,MacKinnonKramer1983}.
Subsequent transfer-matrix studies quantified corrections to scaling and
refined estimates of critical parameters in the three-dimensional
orthogonal universality class
\cite{SlevinOhtsuki1999,SlevinOhtsuki2014,SlevinOhtsuki2018}.

Complementary numerical diagnostics include energy-level statistics
\cite{Shklovskii1993}, multifractal finite-size scaling of critical
wave functions \cite{Rodriguez2010,Rodriguez2011}, and the scaling of
conductance averages and distributions
\cite{SlevinMarkosOhtsuki2001,SlevinMarkosOhtsuki2003}.
Multifractal analyses also compare the three Wigner--Dyson symmetry
classes in three dimensions \cite{UjfalusiVarga2015}.
Boundary conditions can affect critical conductance distributions
\cite{SlevinOhtsukiKawarabayashi2000}. These works motivate testing
several observables, keeping the boundary geometry explicit, and
separating sampling uncertainty from finite-size corrections. Our
conditional boundary comparison and stopped-path memory statistics
address different observables; we do not identify their measured
exponents with the Anderson localization-length exponent.

Numerical studies of the spin quantum Hall transition in planar class-C
networks \cite{Kagalovsky1999} were complemented by exact critical
exponents obtained through a mapping to two-dimensional percolation
\cite{GruzbergLudwigRead1999}. The quantum--classical correspondence
identifies certain averaged
observables of class-C network models with classical routing quantities
\cite{BeamondCardyChalker,Cardy}. Class C belongs to the symmetry
classification of disordered superconducting systems
\cite{AltlandZirnbauer1997}; it differs from the ordinary orthogonal
Anderson class discussed above. The broader classification and numerical
phenomenology are reviewed by Evers and Mirlin \cite{EversMirlin2008}.
Numerical studies of coupled class-C layers also exhibit a
three-dimensional metallic phase \cite{Kagalovsky2004}.
The planar Manhattan lattice has a
different phase structure from its three-dimensional counterpart
\cite{BeamondOwczarekCardy}. Large-scale numerical work on the
three-dimensional diamond lattice \cite{Diamond} concerns a different
orientation and geometry; its numerical transition parameters should not
be transferred to the K-lattice.

The K-lattice literature already studies extended loops, stiffness and
winding observables, and the relation to sigma models
\cite{Nahum2011,Nahum2013}. Brownian finite-loop statistics and
Poisson--Dirichlet macroscopic-loop statistics are developed in
\cite{LoopSoups}. Neither the existence of an extended phase in numerical
experiments nor the familiar three-dimensional return exponent is a new
claim of this paper. Our measurements concern dependence in boundary
reflection maps after their aperture sets are fixed, and the signed
revisit remainder of a corrector decomposition.

Slab concatenation and random-matching comparisons also appear in recent
work on Lorentz mirror transport
\cite{LefevereConductivity,LefevereTasaki}. The latter reference gives
an exactly analyzed hierarchy with uniform undirected interface
matchings. Our comparison retains the physical gluing and the domain and
image of each directed scattering block, and randomizes its internal
bijection. The numerical question is the discrepancy between that
conditional comparison and the original mirror routing; the general
use of random matching as a transport comparison is established prior
work.
Rigorous results for related mirror models address escape probabilities,
cylindrical geometry, and planar Manhattan trajectories
\cite{KS,Ryan,LiManhattan}. These questions also belong to the
broader study of localization in random environments, including discrete
Schr\"odinger operators \cite{LiBernoulli2D,LiZhangBernoulli3D}.

\section{Model and experimental design}\label{sec:model}

\subsection{The four-valent graph and quenched routing}

In coordinates where every graph edge has length one, let
\begin{equation}\label{eq:vertices}
 V=\{x\in\Z^3:x\bmod 2\notin\{(0,0,0),(1,1,1)\}\}.
\end{equation}
At each vertex, precisely two coordinates have the same parity. Their
axes are active, and the graph contains the edges \(x\leftrightarrow
x\pm e_a\) for those two axes. Equivalently, axis \(a\) is active
when the two transverse coordinate parities differ. Orient the active
lines by
\begin{equation}\label{eq:orientation}
 s_1(x)=(-1)^{x_2},\qquad
 s_2(x)=(-1)^{x_3},\qquad
 s_3(x)=(-1)^{x_1}.
\end{equation}
Directions are constant along a line and reversed on nearest parallel
lines. This gives the K-lattice convention of \cite[Section II]{Nahum2013}.
The coordinates are twice those in Cardy's original construction.

A directed state is \((x,a)\), representing arrival at \(x\) along an
active axis \(a\). Sample independent \(\omega_x\sim
\operatorname{Bernoulli}(p)\). If \(\omega_x=0\), the outgoing axis
is \(b=a\); otherwise it is the other active axis. The successor is
\begin{equation}\label{eq:successor}
 (x,a)\longmapsto (x+s_b(x)e_b,b).
\end{equation}
The local pairing is a bijection, so the global successor is a permutation
of directed states. On an infinite trajectory no directed state repeats.
On a finite trajectory the first repeated directed state is the root.
A vertex may be visited twice on its two different incoming axes before
closure; those visits share \(\omega_x\).

For infinite-lattice experiments we start at \(x_0=(1,0,0)\), with
incoming axis \(2\), and write \(\tau\) for its first directed-state
return. Observables are frozen at \(\tau\). Coordinates are never
reduced modulo a period. At \(p=0\) the trajectory is a straight line;
at \(p=1\) it is a six-step hexagon.

\subsection{Slabs and scattering maps}

The boundary experiments use even width \(W\) and thickness \(W\),
with periodic transverse coordinates. Cutting all edges through the two
longitudinal faces gives \(N=W^2/4\) input ports and the same number of
output ports on each face. An incoming boundary trajectory exits at a
unique outgoing port: it cannot enter an internal cycle because a directed
state has a unique predecessor. The scattering map is therefore a
permutation of \(2N\) ports. Internal cycles are allowed and have no
boundary input.

The two directional transmission counts coincide. Indeed, the number of
left inputs not mapped to left outputs equals the number of left outputs
not supplied by left inputs. Denote this common count by \(T\). In
each gluing experiment two independent slabs are joined using the actual
geometric identification of the interface ports. Their concatenation has
thickness \(2W\) and unchanged transverse width \(W\).

All reported transmissions are numbers of directed paths. They omit the
factor of two in the corresponding class-C averaged conductance formula
\cite{Cardy}. The same normalization is used for the physical and
comparison networks.

\subsection{The numerical campaigns}

The boundary campaign contains \(18{,}632\) independent single slabs
and \(3{,}840\) independent pairs for gluing. The latter use widths
\(64,128\) and \(p=0.5,0.7,0.72\). The single-slab campaign includes
additional widths and probabilities for aperture and return-displacement
measurements. The memory campaign contains \(8{,}192\) independent
unwrapped root trajectories, each recorded at three horizons. The annular
campaign contains \(212{,}993\) distinct parameter--seed samples;
Section~\ref{sec:annular} states the cohorts entering each estimate.

Each Monte Carlo unit is an independent environment or pair of
environments. Several channels in one slab, and several observation times
along one trajectory, are not treated as independent samples.

\section{Boundary routing with fixed aperture sets}\label{sec:gluing}

\subsection{An exact conditional comparison}

Let \(A\) and \(B\) be the left and right slab scattering maps,
with transmission counts \(t\) and \(u\). Identify their two interface
port spaces with \(I_+\) and \(I_-\), each of size \(N\), according
to the direction of travel. Define
\begin{align*}
 X&=\{\hbox{outputs in }I_+\hbox{ supplied by left exterior inputs of }A\},\\
 Y&=\{\hbox{inputs in }I_+\hbox{ transmitted through }B\},\\
 Z&=\{\hbox{outputs in }I_-\hbox{ supplied by right exterior inputs of }B\},\\
 U&=\{\hbox{inputs in }I_-\hbox{ transmitted through }A\}.
\end{align*}
Thus \(|X|=|U|=t\) and \(|Y|=|Z|=u\). Put
\begin{equation}\label{eq:overlaps}
 a=|X\cap Y|,\qquad b=|Z\cap U|.
\end{equation}
Randomize the bijection within each of the four scattering blocks of each
slab, retaining its domain and image. The choices are mutually
independent and uniform. For the total transmission count only the two
interface reflection bijections affect the outcome once these sets are
fixed.

\begin{theorem}[Exact comparison with fixed aperture sets]
\label{thm:fixedsets}
Assume \(0<t,u<N\) and \(a<\min(t,u)\). Put
\(K=\min\{N-t,N-u\}\). Let \(s_0=1\) and, for \(0\le k<K\),
\begin{equation}\label{eq:survivalrecursion}
 s_{k+1}=s_k
 \frac{N-t-u+b-k}{N-u-k}
 \frac{N-t-u+a-k}{N-t-k}.
\end{equation}
Then the conditional expected transmission under the preceding
randomization is
\begin{equation}\label{eq:fixedsets}
 H_N(t,u,a,b)=a+(t-a)\sum_{k=0}^{K-1}s_k
 \frac{N-t-u+b-k}{N-u-k}
 \frac{u-a}{N-t-k}.
\end{equation}
After a zero survival factor all subsequent terms are set to zero.
If \(t=0\) or \(u=0\), the value is zero. If \(t=N\) or
\(u=N\), it is \(\min(t,u)\). If \(a=\min(t,u)\), it is \(a\).
\end{theorem}

The proof in Appendix~\ref{app:gluing} reveals uniform reflection
bijections as the path is followed. Each query samples an unused image,
which accounts for the denominators in \eqref{eq:survivalrecursion}.
Although all aperture sets are preserved, the expectation depends only on
their sizes and the two overlaps in \eqref{eq:overlaps}.

For comparison, completely randomizing the two interface matchings yields
the count-only expression
\begin{equation}\label{eq:global}
 F_N(t,u)=tu\sum_{k=0}^{N-\max(t,u)}
 \frac{(N-t)_k(N-u)_k}{(N)_k^2(N-k)},
\end{equation}
when \(tu>0\), and zero otherwise; \((m)_k\) is the falling
factorial. The two comparisons are different. For example,
\(N=2,t=u=1,a=b=0\) gives \(H_N=0\), whereas \(F_N=3/4\).

\subsection{The measured routing deficit}

For each sampled pair, compute its actual \(G\), its exact conditional
\(H=H_N(t,u,a,b)\), and \(F=F_N(t,u)\). Define
\begin{equation}\label{eq:deficits}
 D=H-G,\qquad E=H-F.
\end{equation}
The standard errors in Table~\ref{tab:gluing} are computed from the
paired values of \(D\) and \(E\), rather than from independent-error
propagation of their component means.

\begin{table}[htbp]
\centering\small
\caption{Boundary gluing. The last two columns give paired means with
one standard error. \(m\) is the number of independent slab pairs.}
\label{tab:gluing}
\begin{tabular}{rrrrrrr}
\toprule
\(p\)&\(W\)&\(m\)&\(\overline G\)&\(\overline H\)&\(\overline D\) (SE)&\(\overline E\) (SE)\\
\midrule
.50&64&512&21.39258&21.28829&$-.10428\ (.10020)$&$.01351\ (.02068)$\\
.50&128&256&42.84375&42.77041&$-.07334\ (.20822)$&$-.00613\ (.02771)$\\
.70&64&1024&2.92480&3.11240&$.18760\ (.02746)$&$.00181\ (.00241)$\\
.70&128&512&5.80469&5.93422&$.12954\ (.05471)$&$.00296\ (.00311)$\\
.72&64&1024&1.04688&1.22913&$.18225\ (.01863)$&$.00056\ (.00111)$\\
.72&128&512&1.85938&2.06788&$.20850\ (.03224)$&$.00026\ (.00123)$\\
\bottomrule
\end{tabular}
\end{table}

The \(p=.72,W=128\) deficit has a nominal normal \(95\%\) interval
\([.14531,.27169]\). At both measured widths the positive deficit is
resolved at \(p=.72\), whereas no nonzero deficit is resolved at
\(p=.5\). The fixed-set and count-only comparison means are close in
these ensembles, despite their possible disagreement for individual
terminal sets. This does not make the aperture sets independent.

Only two widths were studied for this intervention. In particular, the
data do not establish a bounded limiting deficit or an asymptotic power
law. The measured effect specifies which dependence must be retained in a
block comparison: internal reflection routing can affect transmission
even when every input and output aperture set is known.

\subsection{Nearby transmitting channels are suppressed}

Let \(\mathcal T\) be the set of incoming transmitting ports on one
face. The null model, conditional on \(|\mathcal T|=T\), is a uniformly
chosen \(T\)-element subset of the \(N\) physical ports. The transverse
port grid is a square torus with \(2N\) nearest-neighbor edges. If
\(J\) counts unordered transmitting nearest-neighbor pairs, then
\begin{equation}\label{eq:neighbor-null}
 \EE_{\mathrm{unif}}[J\mid T]=\frac{2T(T-1)}{N-1}.
\end{equation}
Partition the face into squares of even side \(b\) dividing \(W\).
Each square contains \(m_b=b^2/4\) ports. If \(c_j\) is its
transmitting-port count, then
\begin{equation}\label{eq:aperture-null}
 \EE_{\mathrm{unif}}\left[\sum_jc_j(c_j-1)\,\middle|\,T\right]
 =\frac{T(T-1)(m_b-1)}{N-1}.
\end{equation}
Both formulas follow by summing the probability
\(T(T-1)/(N(N-1))\) for two specified distinct ports to be selected.

We compare the empirical mean observed pair count with the empirical mean
of its conditional null expectation. The two faces are averaged within
each environment before computing a standard error. At \(p=.7,W=128\),
the nearest-neighbor ratio is \(.5663\pm.0452\), using \(2{,}176\)
independent slabs. At \(W=64\) it is \(.6573\pm.0327\), using
\(4{,}352\) slabs. At \(p=.5,W=128\) the ratio is
\(.9393\pm.0145\), using \(640\) slabs.

For \(p=.7,W=128\), the corresponding within-square ratios for
\(b=4,8,16,32,64\) are approximately
\(.651,.722,.820,.915,.970\). The suppression decreases with aperture
size, but these data do not establish a finite correlation length.
The pair statistics measure spatial dependence of the sets; the routing
experiment above holds precisely those sets fixed.

\begin{figure}[htbp]
\centering\includegraphics[width=\textwidth]{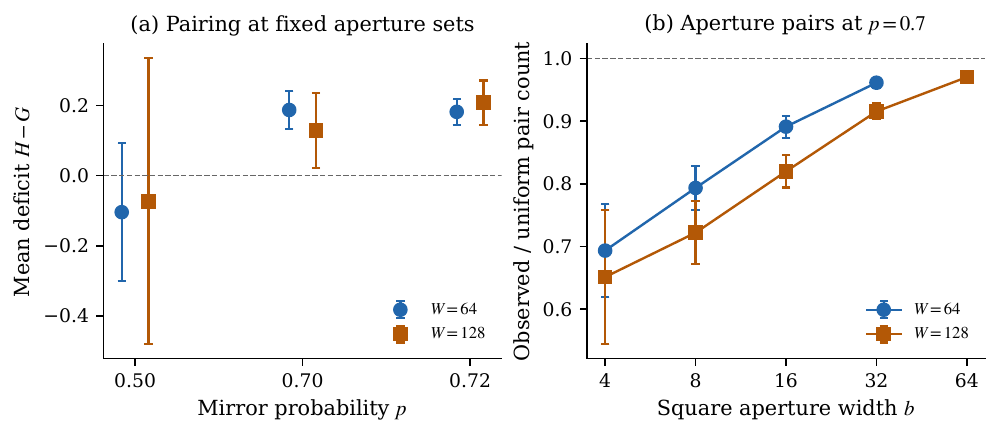}
\caption{Boundary diagnostics. Left: the paired routing deficit
\(H-G\), retaining all terminal sets. Right: ordered within-aperture
pair counts divided by the exact fixed-cardinality null at \(p=.7\).
Error bars are nominal \(95\%\) intervals across independent
environments or environment pairs. The horizontal line at one in the
right panel is the uniform-subset value.}
\label{fig:boundary}
\end{figure}

\section{Age-resolved quenched memory}\label{sec:memory}

\subsection{Charges, stopping, and the observable}

The parity of the vertex and its incoming axis form a twelve-state
space. Appendix~\ref{app:corrector} defines an explicit bounded
corrector \(h\), depending on \(p\in(0,1)\). If \(S_t\) is this
state before processing the vertex at time \(t\), its corrected
increment is
\begin{equation}\label{eq:charge}
 \xi_t=X_{t+1}-X_t+h(S_{t+1})-h(S_t)
       =k(S_t)\left(1-\frac{\omega_{X_t}}p\right),
 \qquad k(S_t)\in\{-1,1\}^3.
\end{equation}
The two incoming states at one vertex have opposite \(k\).
Consequently the corrected charges of its two passages cancel exactly.

Let \(M_n\) sum \(\xi_t\) over first visits with
\(t<n\wedge\tau\). Let \(R_n\) sum the original first-visit
charge over vertices whose second visit also occurs before
\(n\wedge\tau\). Then
\begin{equation}\label{eq:decomposition}
 Y_n:=X_{n\wedge\tau}-X_0+h(S_{n\wedge\tau})-h(S_0)=M_n-R_n.
\end{equation}
Only first visits reveal independent Bernoulli variables, so \(M\) is a
martingale for the exploration filtration. Its predictable trace bracket is
\begin{equation}\label{eq:bracket}
 Q_n=\tr\langle M\rangle_n=\frac{3(1-p)}p F_n,
\end{equation}
where \(F_n\) is the number of freshly processed vertices. No
martingale property is asserted for the selected sum \(R_n\).

For each second visit at time \(t\), let \(s<t\) be the first
visit. Define
\begin{equation}\label{eq:memory-tail}
 R_n^{\ge L}=\sum_{\substack{t<n\wedge\tau:\text{second visit}\\p(t-s)\ge L}}
 \xi_s,\qquad
 D_n^{\ge L}=\sum_{\substack{t<n\wedge\tau:\text{second visit}\\p(t-s)\ge L}}
 |\xi_s|^2.
\end{equation}
The difference \(\Ep|R_n^{\ge L}|^2-\Ep D_n^{\ge L}\) is the
sum of cross-charge inner products. Adaptation of the revisited set to the
revealed environment prevents replacing this difference by zero without
an argument.

\subsection{Measured tails of the signed remainder}

Each probability uses \(2{,}048\) independent paths, observed at
\(n=2^{16},2^{18},2^{20}\). The final-horizon closed counts are
\(0,0,1,5\), in increasing order of \(p\). Closed paths remain in
all averages. The campaign processed \(8{,}583{,}643{,}560\) steps.
Table~\ref{tab:memory} reports estimates of
\begin{equation}\label{eq:normalized-memory}
 B_p(n,L)=\frac{\Ep|R_n^{\ge L}|^2}{p^2\Ep Q_n}.
\end{equation}
These are ratios of ensemble means, not averages of pathwise ratios.

\begin{table}[htbp]
\centering\small
\caption{Normalized signed-square memory tails at \(n=2^{20}\).
Parentheses give one standard error across independent trajectories.}
\label{tab:memory}
\begin{tabular}{rrrr}
\toprule
\(p\)&\(L=0\)&\(L=32\)&\(L=128\)\\
\midrule
.025&$.048414\ (.001461)$&$.018347\ (.000824)$&$.008926\ (.000586)$\\
.050&$.053985\ (.001167)$&$.020077\ (.000492)$&$.009538\ (.000271)$\\
.100&$.061185\ (.001196)$&$.021529\ (.000419)$&$.010658\ (.000231)$\\
.200&$.080896\ (.001636)$&$.026569\ (.000539)$&$.012760\ (.000266)$\\
\bottomrule
\end{tabular}
\end{table}

At \(p=.1,L=128\), the diagonal version
\(\Ep D_n^{\ge L}/(p^2\Ep Q_n)\) is \(.01065798\), compared
with \(.01065837\) for the signed square. Their difference is
\(.00000039\), with standard error \(.00021658\). This agreement
does not extend to an exact independence identity. For example, at
\(p=.2,L=0\) the signed square exceeds its diagonal budget by
\(.004266\pm.001633\), approximately \(5.6\%\) of that budget.

The old-tail and full-tail values are not an additive allocation of
variance. The squared norm of the complete remainder also contains the
cross term between its young and old parts.

\subsection{Finite-horizon fits}

At horizon \(H=pn\), two visits with age \(\ell\) have a
placement factor proportional to \(H-\ell\). A stationary age
density proportional to \(\ell^{-1-\beta}\) would therefore give,
up to an amplitude,
\begin{equation}\label{eq:fit-shape}
 f_\beta(L,H)=L^{-\beta}-H^{-\beta}
 -\frac{\beta}{1-\beta}
    \left(H^{-\beta}-\frac{L^{1-\beta}}H\right),
 \qquad 0<L<H.
\end{equation}
For \(\beta=1/2\), this becomes
\begin{equation}\label{eq:half-shape}
 f_{1/2}(L,H)=L^{-1/2}-2H^{-1/2}+L^{1/2}/H.
\end{equation}
We fit an arbitrary amplitude times \eqref{eq:fit-shape} to the signed
square in \eqref{eq:normalized-memory}, by least squares of logarithms
at \(L=32,64,128,256,512,1024\). Applying the placement model to a
signed-square statistic is a modeling assumption: cross-charge
correlations need not have the same horizon dependence as counts.

\begin{table}[htbp]
\centering\small
\caption{Finite-horizon corrected tail exponents. Intervals use
\(500\) whole-trajectory bootstrap samples and quantify sampling
uncertainty within the stated fit model.}
\label{tab:exponents}
\begin{tabular}{rrrr}
\toprule
\(p\)&Estimate \(\widehat\beta\)&Bootstrap \(95\%\) interval&Fit starting at \(L=64\)\\
\midrule
.025&.5038&$[.4238,.5940]$&.5158\\
.050&.5000&$[.4681,.5295]$&.4969\\
.100&.4999&$[.4802,.5171]$&.5044\\
.200&.4940&$[.4741,.5115]$&.4872\\
\bottomrule
\end{tabular}
\end{table}

The intervals do not account for all finite-age corrections or uncertainty
in the fit model. The observed consistency with \(1/2\) motivates a
uniform memory estimate; it is not a proof of an asymptotic exponent.

\begin{figure}[htbp]
\centering\includegraphics[width=\textwidth]{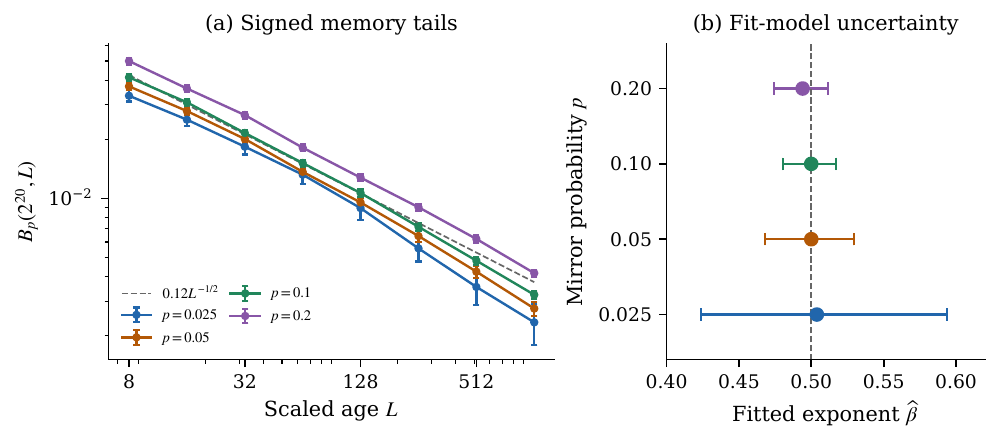}
\caption{Age-resolved memory at \(n=2^{20}\). Left: the normalized
squared norm of the signed old-charge sum; bars are nominal \(95\%\)
intervals and the dashed power law is a visual guide. Right: the
finite-horizon corrected exponents with whole-path bootstrap intervals.
The vertical line is \(\beta=1/2\).}
\label{fig:memory}
\end{figure}

\subsection{Displacement moments without survival conditioning}

Put \(Z_n=X_{n\wedge\tau}-X_0\). Since \(Z_n=0\) once a finite
cycle completes, Cauchy--Schwarz gives
\begin{equation}\label{eq:moment-ratio}
 \Pp(\tau>n)\ge
 \frac{(\Ep|Z_n|^2)^2}{\Ep|Z_n|^4}.
\end{equation}
The final-horizon plug-in ratios are \(.526,.559,.548,.546\) for the
four probabilities above, with standard errors between \(.009\) and
\(.012\). The ratios are not unbiased estimators and are not confidence
bounds on infinite survival. Appendix~\ref{app:criteria} states the
uniform moment estimates that would make \eqref{eq:moment-ratio} a
delocalization argument.

The resampled reference has covariance eigenvalue proportions \(4,4,1\),
as shown in Appendix~\ref{app:corrector}. A centered Gaussian with this
covariance has moment ratio \(27/49\), which is a useful comparison
scale. We make no unconditional Gaussian-limit claim for the quenched
model: finite root loops have positive probability and contribute an atom
at zero after diffusive rescaling.

\section{Annular continuation}\label{sec:annular}

Let $(X_n,a_n)$ denote the directed orbit started from
$(x_0,a_0)=((1,0,0),2)$, in doubled coordinates, and let
\[
 \tau=\inf\{n\geq1:(X_n,a_n)=(x_0,a_0)\},\qquad
 \sigma_R=\inf\{n\geq0:\|X_n-x_0\|_\infty\geq R\}.
\]
As usual, an infimum over the empty set is infinite. Define
\begin{equation}\label{eq:annular-definitions}
 A_R=\{\sigma_R<\tau\},\qquad
 \delta_p(R)=\Pp(A_{2R}^{\mathrm c}\mid A_R).
\end{equation}
Thus $\delta_p(R)$ is the probability that a root orbit closes before
reaching radius $2R$, conditional on having reached radius $R$.
This first-exit observable measures continuation across a spatial scale;
it differs from survival to a prescribed time.

We simulated $8{,}192$ roots through radius $512$ at each of
$p=0.4,0.6,0.7$, and three independent cohorts of $65{,}536$ roots
through radius $128$ at $p=0.7$. The three broad cohorts contain
$196{,}608$ roots. After removing their overlap with the extended runs,
the pooled $p=0.7$ sample contains $196{,}609$ distinct roots; the
campaign contains $212{,}993$ distinct root-environment samples in total.
Every run resolved either closure or its target-radius exit before the
cap of $2^{22}$ steps: there was no time censoring. In particular, the
radius-$512$ exit counts were respectively $8013$, $7007$, and $4378$
out of $8192$, giving empirical probabilities $0.97815$, $0.85535$, and $0.53442$.
These finite-radius observations do not give positive lower bounds on
the probability of an infinite orbit.

\begin{table}[t]
\centering
\caption{Conditional annular losses at $p=0.7$. Intervals are marginal,
two-sided $95\%$ Clopper--Pearson intervals. The eligible cohort contains
$196{,}609$ roots for $R\leq64$ and $8192$ roots for $R\geq128$;
the displayed denominators count those reaching the inner radius.}
\label{tab:annular-loss}
\begin{tabular}{@{}rrrr@{}}
\hline
$R$ & Closures / reached $R$ & $\widehat\delta_{0.7}(R)$ & $95\%$ interval\\
\hline
8   & $12797/126966$ & $0.100791$ & $[0.099140,0.102460]$\\
16  & $5895/114169$  & $0.051634$ & $[0.050357,0.052933]$\\
32  & $2281/108274$  & $0.021067$ & $[0.020220,0.021940]$\\
64  & $979/105993$   & $0.009236$ & $[0.008669,0.009831]$\\
128 & $12/4394$     & $0.002731$ & $[0.001412,0.004766]$\\
256 & $4/4382$      & $0.000913$ & $[0.000249,0.002336]$\\
\hline
\end{tabular}
\end{table}

Table~\ref{tab:annular-loss} and Figure~\ref{fig:annular} show a marked
decrease in annular loss. For $64\to128$, the three broad cohorts gave
$316/35316$, $340/35450$, and $323/35227$, independently reproducing a
loss near one percent. Only sixteen closures contribute to the last
two rows, so their apparent faster decay does not determine a tail
exponent. The measured decay is consistent with investigating a
summable continuation loss.\footnote{Brownian-loop descriptions already
suggest finite-loop tails compatible with inverse-radius corrections;
we do not claim the prediction of a new $R^{-1}$ law. See
\cite{LoopSoups}.}

\begin{figure}[t]
\centering
\includegraphics[width=0.78\textwidth]{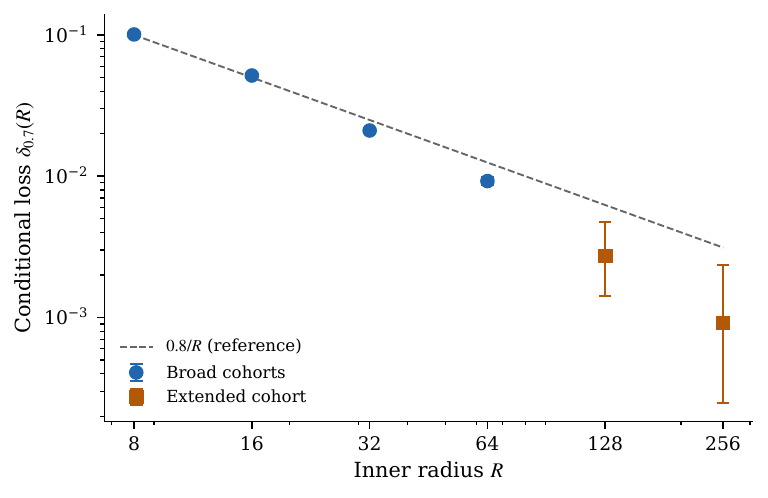}
\caption{Conditional annular loss at $p=0.7$, with marginal $95\%$
intervals. The reference curve $0.8/R$ is not a fitted or proved bound.
The cohort changes after $R=64$, as specified in
Table~\ref{tab:annular-loss}.}
\label{fig:annular}
\end{figure}

\subsection*{First-exit history and replication}
An exploratory split at $R=64$ suggested that slower first exits had
smaller subsequent closure risk. We fixed the cutoffs
$\sigma_{64}/64^2\leq1.01220703125$ and
$\sigma_{64}/64^2>2.408447265625$ before examining the two further cohorts.
The discovery counts were $97/8832$ for fast exits and $43/8829$ for slow
exits. The independent counts were instead $99/8695$ versus $83/9028$,
and $89/8781$ versus $77/8896$. The corresponding slow/fast risk ratios
were $0.807$ and $0.854$, with approximate $95\%$ intervals $[0.604,1.079]$ and
$[0.630,1.157]$. Thus the strong initial association did not replicate.
In the cohort recording exit coordinates, exploratory logistic controls
for squared Euclidean distance and squared displacement along the body
diagonal also left no resolved hit-time association. These tests do not
identify a protective mechanism from the explored history.

\subsection*{A prospective proof criterion}
For $R_k=2^kR_0$, nestedness gives the exact identity
\begin{equation}\label{eq:annular-product}
 \Pp(A_{R_n})
 =\Pp(A_{R_0})\prod_{k=0}^{n-1}
       \bigl(1-\delta_p(R_k)\bigr).
\end{equation}
Consequently, a positive initial exit probability together with
$\delta_p(R_k)<1$ and
$\sum_{k\geq0}\delta_p(R_k)<\infty$ implies positive probability of an
infinite orbit; see Appendix~\ref{app:criteria}. The numerical results
motivate an estimate of this kind for the probability of completing the
\emph{root cycle} before the next radius doubling, averaged over the
first-exit history. They establish neither this all-scale estimate nor
a bound uniform over individual histories.

\section{Computational methods and verification}\label{sec:methods}

\subsection{Environment sampling and directed-state checks}

Infinite-lattice trajectories reveal a Bernoulli variable only at a
first visit and retain it for later passages. An injective integer
encoding stores vertices, their mirror values, and the incoming states
already used. Coordinates are checked against the encoding range.
The exploration stops only on return of the complete root directed
state. A repeated non-root directed state, an inactive arrival axis,
or a third processing of one vertex before closure is an error.

The slab routines follow all boundary inputs through a single frozen
environment. They verify that the resulting boundary map is a
permutation and that the directional transmission counts agree.
Periodic transverse coordinates are used only for the slab experiments.
The root-memory and root-annular experiments have unwrapped coordinates.

Each environment has a recorded seed. A trajectory is the statistical
unit for all memory observables; a slab or slab pair is the unit for
boundary observables. Original and refined batches were pooled only
after checking sample identities. Overlapping root seeds in two annular
batches were counted once; their common initial paths supplied an
additional replay check. The stated sample sizes are the completed
predetermined counts, with no outcome-dependent stopping of a batch.

\subsection{Monte Carlo errors and fits}

For independent paired measurements \(d_1,\ldots,d_m\), the estimated
standard error of their mean is the sample standard deviation divided
by \(\sqrt m\). This convention is used for the two deficits in
\eqref{eq:deficits}. It retains the correlation between physical and
comparison transmissions in the same environment pair.

For a ratio of means \(\widehat r=\overline A/\overline B\), we use
the empirical influence values
\begin{equation}\label{eq:ratio-se}
 \frac{A_i-\widehat r B_i}{\overline B}.
\end{equation}
Their sample standard deviation divided by \(\sqrt m\) gives the
delta-method standard error. Boundary-face measurements are first
averaged within each environment. For memory fits, bootstrap resampling
selects whole trajectories, retaining their joint observations at
different ages. Supplementary percentile intervals for the memory tails
use \(1{,}000\) resamples; fitted-exponent intervals use \(500\).

Finite-size scaling studies of the Anderson transition explicitly test
corrections to scaling and the stability of fits under changes of the
data range \cite{SlevinOhtsuki1999,SlevinOhtsuki2014,Rodriguez2011}.
Here the change of the minimum fitted age in Table~\ref{tab:exponents}
is a limited sensitivity check. It does not replace a systematic
analysis of correction terms or make the bootstrap intervals account
for uncertainty in the scaling ansatz.

The annular intervals are exact binomial intervals for the number of
failures among paths that reached the inner radius. Measurements at
successive radii are dependent and are not treated as independent
observations in a fitted scaling law. Zero observed failures would give
a nonzero upper confidence limit. No infinite-scale probability is
estimated by simply replacing an unobserved tail with zero.

\subsection{Exact controls and independent checks}

The endpoint controls give straight trajectories at \(p=0\) and
elementary hexagons at \(p=1\). A separate implementation matched
\(32\) memory trajectories at all \(96\) common snapshots, including
their fresh and repeated counts and the vectors \(M,R,X\). In the
full memory campaign the maximum absolute residual in
\eqref{eq:decomposition} was \(4.55\times10^{-13}\).
The sums of the recorded age-bin vectors and counts reproduce the
unbinned observables. A \(256\)-trajectory comparison of the annular
simulator reproduced closure, step count, range, and revisit counts
from a separately implemented trajectory routine.

The formula \eqref{eq:fixedsets} was checked against exact enumeration
of \(5{,}060\) pairs of reflection bijections in \(145\) feasible
size/overlap cases with \(2\le N\le5\). These finite checks support
the implementation; the proof in Appendix~\ref{app:gluing} establishes
the general formula. In addition, \(24\) independently sampled
rematchings per physical slab pair agree with the exact conditional
mean within their Monte Carlo uncertainty. The reported primary deficit
uses the exact mean and does not include this additional rematching noise.

The periodic obstruction was checked both in contact coordinates and
directly in physical coordinates. A physical period-four box has
\(96\) directed states; they form two cycles of length \(48\), each
with zero unwrapped displacement. Appendix~\ref{app:geometry} supplies
the complete finite construction and the connectivity argument, so its
validity does not depend on a simulation extrapolation.

\subsection{Reproducibility and scope}

The numerical supplement contains the model implementations, recorded
seeds, per-environment and per-trajectory observations, analysis scripts,
and the exact finite checks. The manuscript tables are fixed summaries
of these observations. The accompanying \LaTeX{} package includes all
figures and uses no external numerical input files when typeset.

Monte Carlo intervals describe sampling uncertainty in the specified
finite experiments. They are not deterministic error enclosures and
do not include every finite-size or modeling error. The boundary
comparison preserves aperture sets but changes the pairing law.
The memory exponent fits assume the finite-horizon form
\eqref{eq:fit-shape}. These distinctions matter when translating
numerical observations into a mathematical conjecture.

\section{Discussion}\label{sec:discussion}

The results suggest three analytical questions. First, compare physical
reflection with the conditional law retaining its boundary sets. One
possible estimate is
\begin{equation}\label{eq:gluing-target}
 \EE\!\left[H_N(t,u,a,b)-G\right]\le C_p
\end{equation}
for independent slabs, with \(C_p\) independent of width. Two tested
widths do not establish an asymptotic bound. A proof would also need a
supply of transmitting paths, a multiscale construction, and a passage
to infinite-volume escape. The exact \(H_N\) avoids assuming uniform
aperture sets.

Second, the memory experiment motivates the estimate
\begin{equation}\label{eq:memory-target}
 \Ep|R_n^{\ge L}|^2\le
 C p^2(1+L)^{-1/2}\Ep Q_n
\end{equation}
uniformly over a small-density regime and relevant \(n,L\). Separating
the diagonal collision budget from cross-charge inner products may help.
Exact pair cancellation does not control the selection of charges for
cancellation. Moreover, \eqref{eq:memory-target} alone does not prove
delocalization: the stopped fresh-vertex count in the bracket also needs
a growth estimate. Uniform second- and fourth-moment bounds for the
stopped displacement give one sufficient completion
(Appendix~\ref{app:criteria}).

Third, summable annular losses would yield positive escape through an
infinite product. The observed decrease does not control its remaining
tail. The failed replication of the first-exit-time association also
cautions against choosing a history descriptor from one exploratory
cohort.

The contact representation gives a geometric constraint. Straight
vertices are independent open cubic-lattice contacts between all-turn
hexagons. By Theorem~\ref{thm:periodic-obstruction}, a connected contact
graph filling space can support only finite trajectories: contacts can
split as well as merge routing cycles. An iid proof must use more than
contact connectivity or component density.

This study establishes neither a critical probability nor a delocalized
phase. Its boundary comparison, signed-memory diagnostic, and exact
structural constraints identify mechanisms for further analytical work.
Larger widths and time horizons can test their proposed uniformity.
The established Anderson-transition literature suggests examining
several size sequences and distributional observables
\cite{MacKinnonKramer1983,SlevinMarkosOhtsuki2003}; for the present
model, these should be coupled to the routing and memory quantities
that enter prospective proofs.

\section*{Acknowledgments}
The authors used OpenAI's GPT models to assist with mathematical
exploration, the development and review of proofs, the design and
implementation of numerical experiments, drafting and revising the
manuscript, and checking bibliographic information. The authors take full
responsibility for all mathematical statements, numerical results,
proofs, and the final content of this paper.

\clearpage
\appendix

\section{Exact boundary comparison}\label{app:gluing}

We prove Theorem~\ref{thm:fixedsets} by exposing the two reflection
bijections along one transmitted input.  The argument concerns finite
scattering maps and requires no assumption on their realization by a local
mirror configuration.

Identify the two directed interface channel sets with disjoint sets
$\Omega_+$ and $\Omega_-$, each of cardinality $N$.  The former carries
channels from the left slab to the right slab, and the latter carries
channels in the opposite direction.  Write
\begin{itemize}
\item $I\subseteq\Omega_+$ for the images of left-to-right transmission
through the left slab;
\item $T\subseteq\Omega_+$ for the inputs transmitted to the external
right boundary by the right slab;
\item $E\subseteq\Omega_-$ for the inputs transmitted to the external
left boundary by the left slab;
\item $J\subseteq\Omega_-$ for the images of right-to-left transmission
through the right slab.
\end{itemize}
These sets are \((I,T,E,J)=(X,Y,U,Z)\) in the notation of
Section~\ref{sec:gluing}, with \(\Omega_\pm=I_\pm\).
The permutation property of each scattering map implies
\[
 |I|=|E|=t,\qquad |T|=|J|=u.
\]
Set
\[
 a=|I\cap T|,\qquad b=|E\cap J|,
 \qquad m_+=N-t-u+a,\qquad m_-=N-t-u+b.
\]
In particular,
\[
 \max\{0,t+u-N\}\le a,b\le\min\{t,u\}.
\]
Thus $m_+=|\Omega_+\setminus(I\cup T)|$ and
$m_-=|\Omega_-\setminus(E\cup J)|$.

Under the intervention of Theorem~\ref{thm:fixedsets}, the two reflection
maps are independent uniform bijections
\begin{equation}\label{eq:reflectionmaps}
 B:\Omega_+\setminus T\longrightarrow\Omega_-\setminus J,
 \qquad
 A:\Omega_-\setminus E\longrightarrow\Omega_+\setminus I.
\end{equation}
The other two maps in each slab are also rematched within their prescribed
blocks.  Their detailed values do not affect the total number of paths
transmitted through the glued sample: every channel in $I$ is entered
exactly once by the $t$ external inputs that transmit through the left
slab, and reaching any channel in $T$ is successful transmission.
Consequently it is enough to study the maps in~\eqref{eq:reflectionmaps}.

A path starting at $x_0\in I$ succeeds immediately if $x_0\in T$.
Otherwise it follows $B$ to a channel $y_0$.  If $y_0\in E$, it exits
at the external left boundary and fails.  If $y_0\notin E$, it follows
$A$ to a channel $x_1$.  The same rule is then repeated.  For any fixed
bijections $A,B$, a path started in $I$ cannot become trapped in an
interface cycle.  Indeed, $x_0$ is outside the image of $A$.  A repeated
plus channel other than $x_0$ would, by injectivity of $A$ and then of
$B$, force a repetition of its predecessor.  Iterating backwards would
force a later visit to $x_0$, which is impossible.  A repeated minus
channel similarly forces a repeated plus channel.  All channels visited
before absorption are therefore distinct, and the path terminates after
finitely many queries.

\begin{proof}[Proof of Theorem~\ref{thm:fixedsets}]
There is nothing to transmit if $t=0$ or $u=0$, giving
$H_N(t,u,a,b)=0$.  If $t=N$, then $I=\Omega_+$ and $E=\Omega_-$.  The $u$ inputs in $T$ succeed directly,
whereas all other inputs exit to the left after their first reflection.
Thus $H_N=u$.  If $u=N$, every channel in $I$ succeeds directly, and
$H_N=t$.  These are the endpoint values stated in the theorem.

Suppose henceforth that $0<t,u<N$.  If $a=t$, all $t$ injected
channels succeed directly, and the formula holds because its summation
term is multiplied by $t-a=0$. If $a=u<t$, then $T\subset I$;
the image of $A$ excludes $I$, so only the $a$ direct successes are
possible. We may therefore assume $a<\min(t,u)$ and fix
one of the $t-a$ channels $x_0\in I\setminus T$.  A completed unsuccessful round means that the
path has queried $B$ and then $A$, without yet reaching either external
boundary.  Let $s_k$ be the probability that $k$ such rounds have been
completed, starting with $s_0=1$.

Condition on any particular history of $k$ completed unsuccessful
rounds having positive probability.  This history has exposed exactly
$k$ distinct pairs of $B$ and $k$ distinct pairs of $A$.  The $k$ exposed
images of $B$ all lie in
$\Omega_-\setminus(E\cup J)$, whereas the $k$ exposed images of $A$
all lie in $\Omega_+\setminus(I\cup T)$.  Conditional on these pairs,
each unexposed bijection remains uniform on its remaining domain and
image.  This assertion follows directly by counting extensions: a
partial bijection with $k$ prescribed pairs has $(d-k)!$ extensions
when its full domain has cardinality $d$.  Conditioning additionally
on the given path history imposes no further restriction on the
unexposed pairs, since the history is determined by the prescribed
pairs and the four fixed sets.

The next plus channel has not previously been queried.  Its image
under $B$ is therefore uniform among $N-u-k$ unused minus channels.
Exactly $m_--k$ of these avoid $E$.  Conditional on avoiding $E$, the
ensuing query of $A$ is also at an unqueried channel, and its image is
uniform among $N-t-k$ unused plus channels.  Precisely $u-a$ of those
images lie in $T$: no such image has been used in an unsuccessful
round.  The remaining $m_+-k$ images avoid $T$.  It follows that the
probability of success on this round is
\begin{equation}\label{eq:successround}
 s_k\frac{m_--k}{N-u-k}\frac{u-a}{N-t-k},
\end{equation}
and the probability of completing one further unsuccessful round is
\[
 s_{k+1}
 =s_k\frac{m_--k}{N-u-k}\frac{m_+-k}{N-t-k}.
\]
This is exactly~\eqref{eq:survivalrecursion}.  The calculation applies
to every positive-probability history with $k$ rounds, so it also gives
the stated unconditional recursion.  Once a continuation numerator
vanishes, all subsequent continuation probabilities are set to zero.

A success counted by~\eqref{eq:successround} requires both a new query
of $B$ and a new query of $A$.  Hence only
\[
 0\le k< K,\qquad K=\min\{N-t,N-u\},
\]
can contribute.  This remains true if a path can make one final query
of $B$ after all images of $A$ have been used: that last query can only
exit to the left, and cannot contribute a success.  Summing
\eqref{eq:successround} gives the success probability for any fixed
$x_0\in I\setminus T$.  There are $a$ direct successes and $t-a$
such remaining inputs.  Linearity of expectation, which does not
require different input paths to be independent, gives
\[
 H_N(t,u,a,b)
 =a+(t-a)\sum_{k=0}^{K-1}
 s_k\frac{N-t-u+b-k}{N-u-k}\frac{u-a}{N-t-k}.
\]
This is~\eqref{eq:fixedsets} and completes the proof.
\end{proof}

The proof shows that the conditional expectation depends on the four
fixed aperture sets only through $N,t,u,a,b$.  It does not say that the
physical reflection maps, conditional on those sets, are uniform.
Nor need the rematched boundary maps be realizable by mirrors inside
the slabs.  These distinctions are essential when using the exact
comparison to interpret the numerical experiment.

For completeness, the count-only comparison in \eqref{eq:global}
follows from the same exposure argument. With independent uniform
permutations of the two interface directions, survival through $k$
unsuccessful round trips has probability
\[
 \frac{(N-t)_k(N-u)_k}{(N)_k^2}.
\]
The next forward query is uniform among $N-k$ unused ports, of which
$u$ transmit through the right slab. Multiplying by $u/(N-k)$,
summing over $0\le k\le N-\max(t,u)$, and then summing over the $t$
injected channels gives \eqref{eq:global}. No independence between
different injected trajectories is needed.

\section{Sufficient conditions for infinite trajectories}\label{app:criteria}

Fix a directed root state, and write $X_n$ for the position of its
trajectory.  Let
\[
 \tau=\inf\{n\ge1:\text{the directed state at time $n$ equals the root
 state}\},
\]
with $\inf\varnothing=\infty$.  The directed state includes both the
vertex and the incoming axis.  In particular, $X_\tau=X_0$ whenever
$\tau<\infty$.  Define the displacement frozen at closure by
\[
 Z_n=X_{n\wedge\tau}-X_0.
\]
All expectations in this section include closed trajectories.  They
are not conditioned on survival.

\begin{proposition}\label{prop:stoppedmoments}
For every $n$ with $\mathbb E|Z_n|^4>0$,
\begin{equation}\label{eq:stoppedmomentbound}
 \mathbb P(\tau>n)
 \ge \frac{(\mathbb E|Z_n|^2)^2}{\mathbb E|Z_n|^4}.
\end{equation}
In particular, if constants $c,C>0$ satisfy
\[
 \mathbb E|Z_n|^2\ge cn,
 \qquad
 \mathbb E|Z_n|^4\le Cn^2
\]
for every sufficiently large $n$, then
$\mathbb P(\tau=\infty)\ge c^2/C>0$.
\end{proposition}

\begin{proof}
On $\{\tau\le n\}$, the stopped displacement is zero.  Therefore
Cauchy--Schwarz gives
\[
 \mathbb E|Z_n|^2
 =\mathbb E\bigl[|Z_n|^2\mathbf1_{\{\tau>n\}}\bigr]
 \le (\mathbb E|Z_n|^4)^{1/2}
       \mathbb P(\tau>n)^{1/2}.
\]
Squaring proves~\eqref{eq:stoppedmomentbound}.  Under the additional
hypotheses, its right-hand side is at least $c^2/C$ for all sufficiently
large $n$.  The events $\{\tau>n\}$ decrease to $\{\tau=\infty\}$,
so continuity of probability from above proves the last assertion.
\end{proof}

When $\mathbb E|Z_n|^4=0$, both moments vanish and the ratio may be
assigned the value zero.  More generally, a positive lower limit of
the ratios in~\eqref{eq:stoppedmomentbound} is sufficient.  The same
argument applies to a corrected displacement
\[
 Z_n+h(S_{n\wedge\tau})-h(S_0),
\]
provided its second and fourth moments satisfy the corresponding
bounds: its value also vanishes at a completed root cycle.  A
second-moment estimate for the fresh-visit martingale alone does not
supply either of the hypotheses for the actual displacement.

We next record the radius formulation.  Put
\[
 \sigma_R=\inf\{n\ge0:\lVert X_n-X_0\rVert_\infty\ge R\},
 \qquad A_R=\{\sigma_R<\tau\}.
\]
The channel graph has finitely many directed states in any bounded
box.  For each fixed environment, the routing map on all directed
states is a bijection.  Consequently its root orbit either returns
to the root or visits distinct directed states forever.  In the
second case it leaves every bounded box.

\begin{proposition}\label{prop:annularcriterion}
Fix $R_0>0$, write $R_j=2^jR_0$, and suppose
$\mathbb P(A_{R_j})>0$ for every $j\ge0$.  Define
\[
 \delta_{R_j}
 =1-\mathbb P(A_{R_{j+1}}\mid A_{R_j}).
\]
Then
\begin{equation}\label{eq:annularproduct}
 \mathbb P(\tau=\infty)
 =\mathbb P(A_{R_0})\prod_{j=0}^{\infty}(1-\delta_{R_j}).
\end{equation}
If $\delta_{R_j}<1$ for every $j$ and
$\sum_{j\ge0}\delta_{R_j}<\infty$, the probability in
\eqref{eq:annularproduct} is positive.  In particular, an estimate
\[
 \delta_R\le C R^{-\alpha},\qquad \alpha>0,
\]
at every sufficiently large dyadic radius suffices, provided the
trajectory has positive probability of reaching each finite radius.
\end{proposition}

\begin{proof}
The events $A_{R_j}$ are nested.  The definition of conditional
probability therefore gives, for every $m\ge1$,
\[
 \mathbb P(A_{R_m})
 =\mathbb P(A_{R_0})\prod_{j=0}^{m-1}(1-\delta_{R_j}).
\]
The bijectivity observation above shows that
$\bigcap_{j\ge0}A_{R_j}=\{\tau=\infty\}$: an infinite orbit leaves
every finite box, whereas a finite root cycle has bounded range.
Taking $m\to\infty$ proves~\eqref{eq:annularproduct}.

If the losses are summable, then $\delta_{R_j}\le1/2$ for all
sufficiently large $j$.  For $0\le x\le1/2$,
$\log(1-x)\ge-2x$.  The logarithm of the tail product is thus bounded
below by a finite number.  Each factor in the finite initial product
is positive, proving positivity of the full product.  Finally,
$\sum_j C(2^jR_0)^{-\alpha}$ is finite for $\alpha>0$.
\end{proof}

For this model with $0\le p<1$, reaching any prescribed finite radius
has positive probability: prescribing straight pairings at the
finitely many vertices of an initial directed line segment forces
the required excursion.  Thus the reachability hypothesis of
Proposition~\ref{prop:annularcriterion} causes no additional issue in
the parameter range considered here.  If the power bound holds from
$R_0$ onwards and $CR_0^{-\alpha}\le1/2$, the proof also gives the
explicit estimate
\[
 \mathbb P(\tau=\infty)
 \ge \mathbb P(A_{R_0})
 \exp\left\{-\frac{2CR_0^{-\alpha}}{1-2^{-\alpha}}\right\}.
\]

The losses $\delta_R$ average over the random revealed history at the
first exit from the smaller box.  The criterion requires no estimate
uniform over all possible revealed histories.  On the other hand,
observing decreasing losses at finitely many radii does not establish
the summable bound.  Likewise, finite-time empirical moment ratios
do not establish the uniform moment hypotheses of
Proposition~\ref{prop:stoppedmoments}.  The two propositions specify
analytical conclusions that would suffice; they are not consequences
of the simulations alone.

\section{The corrector and the revisit decomposition}\label{app:corrector}

We give the finite-state calculation underlying the memory observables. All
identities in this appendix concern the cyclic orientation specified in
Section~\ref{sec:model}. We use column vectors, and subscripts on coordinates
are read cyclically modulo three. Let
\[
 \mathcal S=\{(q,a):q\in\{0,1\}^3,\ q_{a+1}\ne q_{a+2},\ a\in\{1,2,3\}\}.
\]
There are twelve states. The parity vectors $(0,0,0)$ and $(1,1,1)$ do not
occur, and each of the other six parity vectors has two active axes. For
$s=(q,a)\in\mathcal S$, write $a_0=a$, and let $a_1$ be the other active axis
at $q$. Define the state maps and their pullback operators by
\begin{equation}\label{eq:corrector-state-maps}
 \Phi_b(q,a)=(q+e_{a_b}\pmod 2,a_b),\qquad
 (P_b f)(s)=f(\Phi_b(s)),\qquad b\in\{0,1\}.
\end{equation}
Here $b=0$ denotes a straight vertex and $b=1$ a turning vertex. The incoming
velocity is
\[
 V(q,a)=(-1)^{q_{a+1}}e_a.
\]
If the scatterer at the current vertex has value $b$, then the displacement
of the next step is $V(\Phi_b(s))$. Indeed, moving along axis $a_b$ does not
change the transverse coordinate determining its directed sign.

Put
\begin{equation}\label{eq:corrector-definitions}
 U=P_1V,\qquad W=P_1^2V,\qquad
 C=\frac{4I-J}{2},\qquad
 k=V-U-W,\qquad h=\frac1p CV+W-2V,
\end{equation}
where $I$ is the $3\times3$ identity matrix and $J$ has every entry equal to
one. In particular, $h$ is a bounded function on $\mathcal S$ for every
fixed $p\in(0,1)$.

\begin{proposition}\label{prop:corrector}
For every $s\in\mathcal S$ and $b\in\{0,1\}$, the corrected increment is
\begin{equation}\label{eq:corrector-increment}
 \xi(s,b):=V(\Phi_b(s))+h(\Phi_b(s))-h(s)
       =k(s)\left(1-\frac bp\right).
\end{equation}
Every coordinate of $k(s)$ is $1$ or $-1$. At a fixed vertex the two possible
incoming states have opposite values of $k$. Consequently, if $B$ is an
independent Bernoulli random variable with parameter $p$, then
\begin{equation}\label{eq:corrector-centered}
 \mathbb E[\xi(s,B)]=0,\qquad
 \mathbb E[\xi(s,B)\xi(s,B)^{\mathsf T}]
   =\frac{1-p}{p}\,k(s)k(s)^{\mathsf T}.
\end{equation}
Two traversals of the same vertex through different incoming axes have
exactly opposite corrected increments when they use the same scatterer.
\end{proposition}

\begin{proof}
The entire calculation can be checked from the following table. In its last
column a triple denotes the coordinates of the vector $k$.
\begin{center}
\renewcommand{\arraystretch}{1.12}
\begin{tabular}{cccccc}
\hline
$q$ & $a$ & $V$ & $U$ & $W$ & $k$ \\
\hline
$(0,0,1)$ & $1$ & $e_1$ & $-e_2$ & $e_3$ & $(1,1,-1)$ \\
$(0,0,1)$ & $2$ & $-e_2$ & $e_1$ & $-e_3$ & $(-1,-1,1)$ \\
$(0,1,0)$ & $1$ & $-e_1$ & $e_3$ & $-e_2$ & $(-1,1,-1)$ \\
$(0,1,0)$ & $3$ & $e_3$ & $-e_1$ & $e_2$ & $(1,-1,1)$ \\
$(0,1,1)$ & $2$ & $-e_2$ & $e_3$ & $-e_1$ & $(1,-1,-1)$ \\
$(0,1,1)$ & $3$ & $e_3$ & $-e_2$ & $e_1$ & $(-1,1,1)$ \\
$(1,0,0)$ & $2$ & $e_2$ & $-e_3$ & $e_1$ & $(-1,1,1)$ \\
$(1,0,0)$ & $3$ & $-e_3$ & $e_2$ & $-e_1$ & $(1,-1,-1)$ \\
$(1,0,1)$ & $1$ & $e_1$ & $-e_3$ & $e_2$ & $(1,-1,1)$ \\
$(1,0,1)$ & $3$ & $-e_3$ & $e_1$ & $-e_2$ & $(-1,1,-1)$ \\
$(1,1,0)$ & $1$ & $-e_1$ & $e_2$ & $-e_3$ & $(-1,-1,1)$ \\
$(1,1,0)$ & $2$ & $e_2$ & $-e_1$ & $e_3$ & $(1,1,-1)$ \\
\hline
\end{tabular}
\end{center}
To make the algebra explicit, applying the state maps
\eqref{eq:corrector-state-maps} to these entries gives
\begin{equation}\label{eq:corrector-operator-identities}
 \begin{aligned}
 P_0V&=V, & P_0U&=-W, & P_0W&=-U,\\
 P_1V&=U, & P_1U&=W,  & P_1W&=-V.
 \end{aligned}
\end{equation}
In each row the three velocities have the form
$V=\sigma e_a$, $U=-\sigma e_b$, $W=\sigma e_c$, where
$\{a,b,c\}=\{1,2,3\}$ and $\sigma\in\{-1,1\}$. Hence
\[
 C(V-U)=\sigma(e_a+e_b-e_c)=V-U-W=k.
\]
For a straight step, substitution of
\eqref{eq:corrector-operator-identities} into
\eqref{eq:corrector-definitions} yields
\[
 V+P_0h-h=V-U-W=k.
\]
For a turning step the same substitution yields
\[
 U+P_1h-h
   =\frac1p C(U-V)+V-U-W
   =-\frac{k}{p}+k.
\]
This proves \eqref{eq:corrector-increment}. The table shows both the stated
coordinate values and the sign reversal between the two rows at every
fixed $q$. Finally,
$\mathbb E(1-B/p)=0$ and
$\mathbb E(1-B/p)^2=(1-p)/p$, proving
\eqref{eq:corrector-centered}. At two passages through the same vertex the
value of $B$ is unchanged, so the sign reversal of $k$ proves exact
cancellation.
\end{proof}

\subsection{Exploration stopped at the first completed orbit}

Fix a deterministic initial incoming state. Let $(X_t,A_t)$ be the incoming
directed state of its trajectory in a fixed environment, and let
\[
 \tau=\inf\{t\geq1:(X_t,A_t)=(X_0,A_0)\},\qquad N_n=n\wedge\tau.
\]
The vertex $X_t$ is processed before the move to $X_{t+1}$. In particular,
the processed vertices at horizon $n$ have time indices $0\leq t<N_n$.
Write $S_t=(X_t\pmod2,A_t)$ for the parity state.

The map taking one incoming directed state to the next is a bijection.
To see this, the incoming axis of a target state determines its preceding
edge and vertex, and the pairing at that preceding vertex determines a
unique preceding incoming edge. A repeated directed state before $\tau$
would therefore, by applying the inverse map, imply a return to the
initial state at an earlier positive time. Thus there are no repeated
directed states before $\tau$. Since a vertex has exactly two incoming
states, it is processed at most twice before closure, and two such visits
use different incoming axes.

Let $I_t$ be the indicator that $t<\tau$ and $X_t$ has not previously been
processed. With
$\xi_t=\xi(S_t,B_{X_t})$ before closure, define
\begin{align}
 M_n&=\sum_{0\leq t<N_n} I_t\xi_t,\label{eq:corrector-fresh-martingale}\\
 R_n&=\sum_{v:\,t_2(v)<N_n}\xi_{t_1(v)},\label{eq:corrector-memory-remainder}
\end{align}
where $t_1(v)<t_2(v)$ are the two processing times of $v$, when both exist.
All these quantities are frozen after $\tau$.

\begin{proposition}\label{prop:corrector-decomposition}
Under the independent Bernoulli environment law, $M_n$ is a
square-integrable martingale for the filtration revealing a scatterer
only when its vertex is first processed. Its predictable quadratic
variation satisfies
\begin{equation}\label{eq:corrector-bracket}
 \langle M\rangle_n
   =\frac{1-p}{p}\sum_{0\leq t<N_n}I_t k(S_t)k(S_t)^{\mathsf T},
 \qquad
 Q_n:=\tr\langle M\rangle_n
   =\frac{3(1-p)}p N_{\mathrm{fresh}}(n),
\end{equation}
and consequently $\Ep|M_n|^2=\Ep Q_n$. Moreover, pathwise,
\begin{equation}\label{eq:corrector-revisit-decomposition}
 X_{N_n}-X_0+h(S_{N_n})-h(S_0)=M_n-R_n.
\end{equation}
On $\{\tau\leq n\}$ the left-hand side is zero and $M_n=R_n$.
\end{proposition}

\begin{proof}
Just before processing $X_t$, the current directed state, the event
$\{t<\tau\}$ and the first-visit indicator $I_t$ are measurable with
respect to the revealed past. On $\{I_t=1\}$ the variable $B_{X_t}$ is
independent of that past and is Bernoulli with parameter $p$. Proposition~
\ref{prop:corrector} therefore gives a conditionally centered increment
and its conditional covariance. The increments are bounded for fixed
$p$, so summing the conditional covariances proves
\eqref{eq:corrector-bracket} and square integrability at every finite
horizon. The trace formula follows from $|k(s)|^2=3$; martingale
orthogonality gives $\Ep|M_n|^2=\Ep Q_n$.

Summing the definition of $\xi_t$ telescopes to the left-hand side of
\eqref{eq:corrector-revisit-decomposition}. Every first visit contributes
to $M_n$. At each second visit Proposition~\ref{prop:corrector} supplies
the negative of that vertex's first charge. These contributions are
exactly $-R_n$. This proves the pathwise identity. At closure the full
incoming state, and hence also its parity state, equals the initial one.
\end{proof}

\subsection{The resampled covariance benchmark}

For comparison only, consider the walk that samples a new independent
Bernoulli scatterer at every encounter, including repeated visits. Its
parity-state transition operator is $P=(1-p)P_0+pP_1$. Both state maps
are permutations, so the uniform measure on $\mathcal S$ is invariant.
The chain is irreducible for $0<p<1$. One way to check the latter statement
directly is to note that $\Phi_1$ has two cycles of length six. One cycle is
\[
 ((0,0,1),1),\ ((0,1,1),2),\ ((0,1,0),3),\
 ((1,1,0),1),\ ((1,0,0),2),\ ((1,0,1),3).
\]
Its complement is the other cycle, and $\Phi_0$ interchanges
$((0,0,1),1)$ and $((1,0,1),1)$, connecting the two cycles in both
directions. The period is two, which does not affect Ces\`aro convergence
of the state distributions.

The twelve-state table gives
\begin{equation}\label{eq:corrector-k-average}
 \frac1{12}\sum_{s\in\mathcal S}k(s)k(s)^{\mathsf T}
 =K:=\begin{pmatrix}
  1&-1/3&-1/3\\
  -1/3&1&-1/3\\
  -1/3&-1/3&1
 \end{pmatrix}.
\end{equation}
In the resampled model every corrected increment is centered, so their
sum is a martingale. Its expected quadratic variation divided by $n$
converges to $((1-p)/p)K$, by the preceding finite-state averaging.
Since $h$ is bounded, Cauchy--Schwarz shows that removing the endpoint
corrector changes the second-moment matrix by $O(\sqrt n)+O(1)$ at fixed
$p$. Thus, for every fixed initial incoming state, the resampled walk
satisfies
\begin{equation}\label{eq:corrector-annealed-covariance}
 \lim_{n\to\infty}\frac1n
 \mathbb E\big[(X_n-X_0)(X_n-X_0)^{\mathsf T}\big]
 =\Sigma_p:=\frac{1-p}{p}K.
\end{equation}
The eigenvalue in the direction $(1,1,1)$ is $(1-p)/(3p)$ and the two
orthogonal eigenvalues are $4(1-p)/(3p)$. This proves the factor-four
anisotropy of the resampled benchmark. Equation~
\eqref{eq:corrector-annealed-covariance} is not asserted for the quenched
mirror trajectory.

\subsection{What the decomposition does not establish}

For a kinetic-age cutoff $L\geq0$, the corresponding part of the memory
remainder is the precisely defined random vector
\[
 R_n^{\geq L}
 =\sum_{v:\,t_2(v)<N_n}
   \ind_{\{p(t_2(v)-t_1(v))\geq L\}}\,\xi_{t_1(v)}.
\]
The choice of vertices in this sum depends on the revealed scatterers and
on the subsequent path. Therefore neither $R_n$ nor $R_n^{\geq L}$ is
shown by Proposition~\ref{prop:corrector-decomposition} to be a
martingale, and their second moments cannot be obtained by summing
individual charge variances. In particular, the proposed uniform
estimate
\[
 \Ep|R_n^{\geq L}|^2\leq C_0p^2L^{-1/2}\Ep Q_n,
 \qquad L\geq1,
\]
with $C_0$ independent of $n,L$ and of sufficiently small $p$, remains a
probabilistic target suggested by the experiments. It is not
a consequence of the exact two-visit cancellation, of the resampled
covariance calculation, or of a small observed density of repeated
vertices. None of the exact identities above, alone, establishes a
positive probability of an infinite quenched trajectory.

\section{Contact geometry and a periodic obstruction}\label{app:geometry}

This appendix gives an exact representation of the model by local
reconnections of elementary hexagons.  Besides providing geometric
checks on the simulations, the representation separates connectivity of
an auxiliary random graph from escape of the routed trajectories.
Throughout, $p$ is the turning probability and $q=1-p$ is the straight
probability.  Coordinates and orientations are those of
Section~\ref{sec:model} and~\eqref{eq:orientation}.

\subsection{The cubic contact graph}

Write $\mathbf 1=(1,1,1)$ and let
\begin{equation}\label{eq:hexagon-labels}
 \mathcal H=\{h\in\Z^3:h_1\equiv h_2\equiv h_3\pmod 2\}.
\end{equation}
For $h\in\mathcal H$, the six corners of $h+[0,1]^3$ other than
$h$ and $h+\mathbf 1$ form a hexagon in the physical graph.  Denote it by
$\mathcal C_h$.

\begin{proposition}[Cubic contact representation]\label{prop:contact}
The all-turn configuration partitions the directed edges of the physical
graph into the hexagons $\mathcal C_h$, $h\in\mathcal H$.
The graph whose vertices are these hexagons and whose edges are their
contacts is isomorphic to the nearest-neighbor graph on $\Z^3$.
An explicit isomorphism is
\begin{equation}\label{eq:contact-transform}
 U(h)=\left(\frac{h_1+h_2}{2},\frac{h_1+h_3}{2},
                  \frac{h_2+h_3}{2}\right).
\end{equation}
Declare a contact edge open when its physical vertex is straight.
Under the independent mirror measure, the resulting contact graph is
independent bond percolation of parameter $q$ on $\Z^3$.
Every physical trajectory is confined to one open contact component.
A finite component with $n$ hexagons contains exactly $6n$ physical
directed edges.
\end{proposition}

\begin{proof}
For $h=0$, the orientation gives the directed hexagon
\begin{equation}\label{eq:even-hexagon}
 (1,0,0),(1,1,0),(0,1,0),(0,1,1),(0,0,1),(1,0,1),(1,0,0).
\end{equation}
Translation by a vector with all coordinates even preserves the
orientation.  Translation by a vector with all coordinates odd reverses
all arrows, so the corresponding hexagon is traversed in the opposite
order.  At each of its corners the path changes coordinate axis.

For completeness, each physical edge belongs to exactly one of these
hexagons.  Suppose the edge is parallel to coordinate axis $a$, and
let its endpoint with smaller $a$-coordinate have that coordinate $t$.
The lower corner of a unit cube containing the edge must have
$h_a=t$.  For each transverse coordinate $b$, exactly one of the two
choices $h_b=v_b$ and $h_b=v_b-1$ has the same parity as $t$.
This selects a unique $h\in\mathcal H$.  The edge lies among the six
edges of $\mathcal C_h$, because the two transverse coordinates of an
active axis have different parities.  This proves the partition.

At a physical vertex $v$, let $r\in\{0,1\}^3$ be its coordinatewise
parity vector.  Since $v$ belongs to the physical graph, $r$ is neither
$0$ nor $\mathbf 1$.  Exactly two of the selected cubes have $v$ as a
hexagon corner, with lower corners
\begin{equation}\label{eq:contact-endpoints}
 h_0=v-r,\qquad h_1=v-(\mathbf 1-r).
\end{equation}
Their difference belongs to
\begin{equation}\label{eq:contact-vectors}
 h_1-h_0=2r-\mathbf 1
 \in\{\pm(1,1,-1),\ \pm(1,-1,1),\ \pm(-1,1,1)\}.
\end{equation}
Conversely, any pair in $\mathcal H$ with one of these differences has
exactly one physical contact vertex.

The inverse of~\eqref{eq:contact-transform} is
\begin{equation}\label{eq:contact-inverse}
 h=(u_1+u_2-u_3,\ u_1-u_2+u_3,\ -u_1+u_2+u_3).
\end{equation}
It maps $\Z^3$ bijectively onto $\mathcal H$, and $U$ maps the three
positive vectors displayed in~\eqref{eq:contact-vectors} to the three
standard coordinate vectors.  This proves the graph isomorphism and
the bijection between physical vertices and contact edges.

At a turn vertex, each incoming edge continues around its own
hexagon.  At a straight vertex, the two continuations are exchanged,
and each incoming edge continues along the other hexagon.
Independence of the physical vertex choices therefore gives the stated
independent bond measure.  A trajectory can change its hexagon label
only across an open contact edge, which proves confinement.
Finally, the hexagons are disjoint as collections of directed edges
and each has six such edges.
\end{proof}

In particular, if $q$ is below the critical probability for bond
percolation on the cubic lattice, every physical trajectory is finite
almost surely.  The converse does not follow from the contact
representation, even at the deterministic level, as the construction
below demonstrates.

\subsection{Circuits inside a contact component}

\begin{proposition}[Cycle-surplus bound]\label{prop:cycle-surplus}
Let a finite open contact component have $n$ vertices and $m$ edges,
and let $k$ be the number of physical routing loops that it contains.
Writing $\beta=m-n+1$ for its cycle surplus, there is an integer $g\geq0$
such that
\begin{equation}\label{eq:cycle-surplus-identity}
 k=\beta+1-2g.
\end{equation}
Consequently,
\begin{equation}\label{eq:cycle-surplus-bound}
 1\leq k\leq\beta+1,\qquad
 k\equiv\beta+1\pmod 2,
\end{equation}
and its longest routing loop has length at least $6n/(\beta+1)$.
If the contact component is a tree, it contains exactly one routing
loop, of length $6n$.
\end{proposition}

\begin{proof}
Start with all turns on the $n$ hexagons, giving $n$ disjoint cycles.
Changing one contact to straight exchanges the two outgoing
successors assigned to its incoming edges.  Such an exchange joins
two cycles when its incoming edges lie on different cycles, and
splits one cycle into two otherwise.  This follows by cutting the two
successor arrows and reconnecting their four ends.

The exchanges at distinct physical vertices act on disjoint pairs of
outgoing edges, so their order may be chosen arbitrarily.  First
perform the exchanges on a spanning tree of the component.  At each
step, the new edge joins two components of the previously added
forest.  By confinement, its two incoming edges belong to different
routing cycles.  Thus all $n-1$ of these exchanges are joins, leaving
one cycle.  There remain $\beta$ exchanges.  If $g$ of them are joins,
the other $\beta-g$ are splits, and hence
$k=1-g+(\beta-g)$.  This proves the identity and its consequences,
using the total length $6n$ from Proposition~\ref{prop:contact}.
\end{proof}

The following infinite-volume observation identifies one situation in
which contact connectivity does suffice.

\begin{proposition}[Infinite tree components]\label{prop:infinite-tree}
Every physical orbit in an infinite tree contact component is infinite.
\end{proposition}

\begin{proof}
Suppose that a physical orbit is finite.  Its changes of hexagon label
form a closed walk in the contact tree.  For any contact edge crossed
by this walk, deleting that edge separates the tree into two parts,
so the walk crosses it equally often in both directions.
At a given open contact, there are only two directed physical
crossings, one in each direction.  A routing orbit uses neither more
than once before closure, because its successor map is a bijection.
It therefore uses both whenever it uses either.

We claim that if the orbit contains one edge of a hexagon, it contains
all six.  Follow the directed cyclic order of that hexagon.  At a
closed contact, the next edge is on the orbit immediately.  At an
open contact, the orbit crosses to the neighboring hexagon; the
opposite crossing, which also belongs to the orbit by the preceding
paragraph, returns along precisely the next edge of the original
hexagon.  This proves the claim successively around the hexagon.
It also shows that every open contact incident to that hexagon is
used.  Repeating the argument over the connected component forces
the finite orbit to contain every hexagon in an infinite component,
a contradiction.
\end{proof}

\subsection{An explicit periodic configuration}

We now give the construction for
Theorem~\ref{thm:periodic-obstruction}.  The contact coordinates are
denoted by $u$, and $e_1,e_2,e_3$ are their standard basis vectors.
Specify an undirected contact bond by its positive endpoint
description $\{u,u+e_j\}$, with $j\in\{1,2,3\}$.
The configuration is periodic under $2\Z^3$ in contact coordinates.
Its open bonds are exactly those in Table~\ref{tab:periodic-bonds},
together with their $2\Z^3$ translates; all other bonds are closed.

\begin{table}[htbp]
\centering
\caption{Open positive contact directions in the periodic configuration.
Each row specifies all $j$ for which $\{u,u+e_j\}$ is open.}
\label{tab:periodic-bonds}
\begin{tabular}{c|c@{\qquad}c|c}
\hline
$u\bmod 2$ & $j$ & $u\bmod 2$ & $j$\\
\hline
$(0,0,0)$ & $1,2$   & $(1,0,0)$ & $1,2$\\
$(0,0,1)$ & $1,2,3$ & $(1,0,1)$ & $1$\\
$(0,1,0)$ & $1,2$   & $(1,1,0)$ & $3$\\
$(0,1,1)$ & $2$     & $(1,1,1)$ & $1$\\
\hline
\end{tabular}
\end{table}

To verify the routing, it is convenient to label the six physical
vertices on a hexagon by its contact ports.  A port $+j$ joins the
hexagon at $u$ to that at $u+e_j$; a port $-j$ joins it to that at
$u-e_j$.  The directed port orders, indexed by $k=0,\ldots,5$, are
\begin{equation}\label{eq:contact-port-orders}
 (d_0(u),\ldots,d_5(u))=
 \begin{cases}
 (-3,1,-2,3,-1,2),&u_1+u_2+u_3\text{ even},\\
 (-3,2,-1,3,-2,1),&u_1+u_2+u_3\text{ odd}.
 \end{cases}
\end{equation}
Indeed, the even order follows from~\eqref{eq:even-hexagon} and
\eqref{eq:contact-endpoints}; the odd order reverses it.
The parity of the coordinates of $h$ in~\eqref{eq:contact-inverse}
is the parity of $u_1+u_2+u_3$, so these are the required orders
throughout the contact lattice.

A routing state $(u,k)$ represents the physical directed edge ending
at port $d_k(u)$ of the hexagon at $u$.  Put
$a(d)=\operatorname{sgn}(d)e_{|d|}$.  If the contact at this port is
closed, the successor state is $(u,k+1)$, with port indices taken
modulo six.  If it is open, put $u'=u+a(d_k(u))$ and let $k'$ be the
unique index satisfying $d_{k'}(u')=-d_k(u)$.  The successor is then
$(u',k'+1)$.  Thus, on the infinite lattice,
\begin{equation}\label{eq:contact-successor}
 T(u,k)=
 \begin{cases}
 (u,k+1),&\text{the contact is closed},\\
 (u+a(d_k(u)),k'+1),&\text{the contact is open}.
 \end{cases}
\end{equation}
This rule includes the physical step from the selected contact to
the next port, so one iteration corresponds to one physical edge.

\begin{proof}[Proof of Theorem~\ref{thm:periodic-obstruction}]
First, the infinite open contact graph is connected.
Its quotient modulo $2\Z^3$ is connected: the successive residues
\[
 000, 100, 110, 111, 011, 001
\]
are joined by open bonds, the residue $010$ is joined to $000$,
and $101$ is joined to $001$.
Here and below a three-digit string abbreviates a vector in
$\{0,1\}^3$.
In addition, the three paths starting at $0$ with respective step
sequences
\[
 (e_1,e_1),\qquad (e_2,e_2),\qquad
 (e_1,e_2,e_3,e_1,e_2,e_3)
\]
consist entirely of open bonds.  Their endpoint differences are
$2e_1$, $2e_2$, and $2(e_1+e_2+e_3)$, which generate $2\Z^3$.
By periodicity and
reversal of these paths, the component of $0$ contains every
$2\Z^3$ translate of $0$.  Combining them with a lift of a path in
the connected quotient shows that every contact vertex belongs to
this component.

It remains to verify that all physical orbits are finite.
For a residue $\bar u\in\{0,1\}^3$, enumerate the 48 quotient
routing states by
\begin{equation}\label{eq:periodic-state-number}
 s(\bar u,k)=6(4\bar u_1+2\bar u_2+\bar u_3)+k,
 \qquad 0\leq k\leq5.
\end{equation}
Substitution of Table~\ref{tab:periodic-bonds} into
\eqref{eq:contact-successor} gives the following single cycle,
with the final state followed by $0$:
\begin{equation}\label{eq:periodic-fortyeight-cycle}
\begin{aligned}
(&0,10,18,19,20,8,32,33,34,6,7,21,\\
 &22,42,40,12,13,3,4,24,25,39,43,44,\\
 &45,46,47,23,11,35,30,31,9,1,27,28,\\
 &29,5,17,41,36,37,38,26,2,14,15,16).
\end{aligned}
\end{equation}
Every integer from $0$ through $47$ occurs exactly once.
The net displacement in the infinite contact lattice must also be
checked, since a quotient cycle alone would not imply closure.
For transparency, the states at which this cycle crosses an open
contact in each signed direction are listed in
Table~\ref{tab:periodic-crossings}.  Each coordinate has the same
number of positive and negative crossings.  The net contact
displacement is therefore zero.

\begin{table}[htbp]
\centering
\caption{States in~\eqref{eq:periodic-fortyeight-cycle} that cross an
open contact.  All unlisted states remain on their current hexagon.}
\label{tab:periodic-crossings}
\begin{tabular}{c|l|l}
\hline
Coordinate & Positive crossings & Negative crossings\\
\hline
$1$ & $1,11,17,29,31,47$ & $4,8,22,26,34,40$\\
$2$ & $5,7,13,23,25$    & $2,10,16,20,38$\\
$3$ & $9,39$            & $0,42$\\
\hline
\end{tabular}
\end{table}

Consequently, every lift of this quotient cycle returns after 48
steps to its initial contact label and port.  It returns to the
same physical directed edge, since the contact-state description is
bijective.  No earlier return is possible, because the 48 quotient
states are distinct.  Every physical directed edge is a lift of one
of these states, so every physical orbit has length exactly 48.

Finally, the induced physical mirror assignment is periodic.
Adding $4e_a$ to a physical vertex adds an even vector to both of its
contact coordinates in~\eqref{eq:contact-endpoints}--
\eqref{eq:contact-transform}.  Thus the physical assignment is
invariant under $4\Z^3$, completing the construction.
\end{proof}

The configuration has 13 open contact bonds per period cell, out of
24 positive bonds.  Its role is deterministic: it disproves the
implication that an infinite open contact component must contain an
infinite routed trajectory.  It does not give a positive-probability
event for the infinite independent mirror measure, and therefore
does not contradict delocalization in that measure.  Any proof based
on the cubic contact representation must also control how the local
reconnections distribute the component's edges among routing loops.

\end{document}